\documentclass[journal,draftclsnofoot,onecolumn]{IEEEtran}%%Mojtaba
\usepackage{amsmath,amssymb,graphicx,epsfig,cite,fancyhdr,amsthm}
\usepackage{subfigure}
\usepackage{url,hyperref}
\usepackage{tikz}
\usepackage{multicol}
\usepackage{bm}
\usetikzlibrary{calc}
\usepackage{pdfsync}
\usepackage{multirow}% for table
\usepackage{flushend}%to make columns with same length (end page)

\makeatletter
\newcommand*\dashline{\rotatebox[origin=c]{90}{$\dabar@\dabar@\dabar@\dabar@\dabar@\dabar@\dabar@$}}

\makeatletter
\newcommand*{\rom}[1]{\expandafter\@slowromancap\romannumeral #1@}
\makeatother

\makeatother

\newtheorem{thm}{Theorem}[]
\newtheorem{cor}{Corollary}

\theoremstyle{remark}
\newtheorem{rem}{Remark}

\theoremstyle{definition}
\newtheorem{defn}{Definition}

\begin{document}

%%%%%%%\title{On the Capacity of the Gaussian Cognitive Z-Interference Channel:\\ Less Noisy and More Capable CGZIC}
\title{The Capacity of More Capable Cognitive \\Interference Channels}

\author{\IEEEauthorblockN{Mojtaba Vaezi}
\IEEEauthorblockA{Department of Electrical and Computer Engineering\\
McGill University, Montreal, Quebec H3A 0E9, Canada\\
Email: mojtaba.vaezi@mail.mcgill.ca
}}


\maketitle

%----------------------------------------------------------------------

\begin{abstract}

We establish the capacity region for a class of discrete memoryless {\it cognitive interference channel} (DM-CIC)
called {\it cognitive-more-capable} channel, and we show that superposition coding is the optimal encoding technique.
This is the largest capacity region for the DM-CIC to date, as the
existing capacity results are explicitly shown to be its subsets.

%Since, superposition coding is the optimal strategy again.

%
%We explore fundamental limits of the discrete memoryless {\it cognitive interference channel} (DM-CIC) with
%two pairs of transmitter-receiver, in which the cognitive
%transmitter non-causally knows the message of the
%primary transmitter. It is known that superposition coding is optimal
%for several classes of cognitive channels,
%including the recently introduced {\it less noisy} DM-CIC \cite{vaezi2012lessnoisy}.
%In this work, we extend the results of less noisy DM-CIC to a broader class of DM-CIC
%known as more capable DM-CIC.
%Similar to the less noisy DM-CIC, %due to the inherent asymmetry of the cognitive channel,
%two different more capable channels are conceivable: the {\it primary-more-capable}
%and {\it cognitive-more-capable} channels. We establish the capacity region for the latter case, i.e., when
%the cognitive receiver is more capable than the primary receiver,
%and we show that superposition coding is the optimal encoding technique.
%This new result is the largest capacity region for the DM-CIC to date, as all
%existing capacity results are explicitly shown to be its subsets.
\end{abstract}


\IEEEpeerreviewmaketitle

\section{Introduction}\label{sec:intro}

%{\let\thefootnote\relax\footnotetext{
%The author is with the Department of Electrical and Computer
%Engineering, McGill University, Montreal, QC H3A 0E9, Canada (e-mail: 	
%mojtaba.vaezi@mail.mcgill.ca). 	
%}}
Wireless communication systems are limited in performance and
capacity by interference.
The capacity region of the simplest interference channel \cite{Carleial}, i.e.,
the two-user interference channel, is still an open problem despite having been studied for several decades.
It is known only for a few special classes of interference channels, e.g.,
the strong and very strong interference \cite{Kramer}.

With ever-increasing demand for radio spectrum, improving the
spectral utilization in wireless communication systems is unavoidable.
Cognitive radio is recognized as a key enabling technology for this purpose \cite{haykin2005cognitive}.
Owing to the nodes which can sense the environment and
adapt their strategy based
on the network setup, cognitive radio technology is aimed
at increasing the spectral efficiency in wireless communication systems.
With this development in technology, networks with cognitive users are gaining prominence.
Such a communication channel can be modeled by interference channel
with cognition, which is simply known as the {\it cognitive channel} \cite{Devroye}.

Most of the recent work on the cognitive interference channel has focused on
the two user channel with one cognitive transmitter \cite{Devroye,Jovicic-Viswanath,Wu-Vishwanath,Maric2,Maric3,Rini2,vaezi2012lessnoisy}.
In this channel setting, one transmitter, known as the cognitive transmitter, has
non-causal access to the message transmitted by the other
transmitter (the primary transmitter). This is used to model an ``ideal"
cognitive radio.
We study the two-user {\it discrete memoryless} cognitive interference channel, too.
%with two transmitter-receiver pairs known as the primary and the cognitive pairs.
Fundamental limits of this channel have been explored for several years now. However,
the capacity of this channel remains unknown except for some special classes, e.g.,
in the ``weak interference'' \cite{Wu-Vishwanath} and ``strong interference''
\cite{Maric2} regimes.

%Capacity of the Gaussian cognitive interference channel (GCIC) is
%known at low interference \cite{Jovicic-Viswanath} and \cite{Wu-Vishwanath},
%as well as strong very interference  \cite{Maric1}.
%Besides, capacity of Gaussian cognitive Z-interference channel (GCZIC) in
%which the primary receiver is interfered by the cognitive
%transmitter, is known for several ranges of interference gain
%\cite{vaezi2011capacity}, \cite{JiangZ}, \cite{vaezi2011superposition}. While at low interference
%dirty paper coding \cite{Costa} is capacity-achieving scheme, at high
%interference superposition coding is the optimal technique.
%For the discrete memoryless channel, capacity
%is known for ``very strong interference'' \cite{Maric1}, in which both receivers can
%decode both messages, ``weak interference'' \cite{Wu-Vishwanath}, and ``better cognitive decoding'' \cite{Rini2} regimes.
%While independent coding and superposition coding, respectively, are capacity-achieving for the first and second cases,
%the last case uses rate-splitting, Gel'fand-Pinsker coding (binning), and superposition coding.

Inspired by the concept of the less noisy broadcast channel (BC) \cite[Capter 5]{ElGamal2011network}, the author introduced the notion
of {\it less noisy} DM-CIC in \cite{vaezi2012lessnoisy}. Because of the inherent
asymmetry of the cognitive channel, two different less noisy channels are distinguishable;
these are dubbed the {\it primary-less-noisy} and {\it cognitive-less-noisy} DM-CIC.
In the former, the primary receiver is
less noisy than the secondary receiver, whereas it is the opposite in the latter.
In this paper, we extend the work on
{\it less noisy} DM-CIC \cite{vaezi2012lessnoisy} to the {\it more capable} DM-CIC.
The notion of more capable  DM-CIC first appeared in \cite{vaezi2011capacity},
in which the primary receiver is more capable than the secondary one.
We observe that, similar to the less noisy DM-CIC, two different more capable
cognitive channels are conceivable: the {\it primary-more-capable} and {\it cognitive-more-capable} DM-CIC.
The former was studied in \cite{vaezi2011capacity}; the latter
is the subject of study in this work.

The main contribution of this paper is to establish a new capacity result
for the DM-CIC, i.e., capacity
region for the cognitive-more-capable DM-CIC. To this end, we first
propose a new outer bound for this channel; the outer bound
is developed from the outer bound introduced in \cite[Theorem 3.2]{Wu-Vishwanath}.
We then show that this outer bound is the same as an inner bound
which is based on superposition coding. Therefore, we characterize
the capacity region for the cognitive-more-capable DM-CIC and prove
that {\it superposition coding} is the capacity-achieving technique.

In the second part of this paper, we prove that this new capacity result is the
``largest'' capacity region for the DM-CIC to date. To prove this, we explicitly show that the existing
capacity results are subsets of this new capacity region. In fact, the capacity of
the cognitive-more-capable DM-CIC reduces to the capacity regions
in the so-called ``weak interference'' \cite{Wu-Vishwanath}, %\cite{Wu-Vishwanath,Rini2,vaezi2012comments},
``strong interference'' \cite{Maric2},
and ``less noisy" \cite{vaezi2012lessnoisy} regimes once the corresponding channel conditions are satisfied.
Finally, the relation among different capacity results of the DM-CIC is clarified in light of this work and \cite{vaezi2012comments}.
The analysis we provide in this paper sheds more light on the existing capacity results of the DM-CIC;
it makes clear how superposition coding is the capacity achieving technique in the previous results, too.
%Finally, We evaluate the capacity region for the corresponding Gaussian cognitive interference channel (GCIC)
%at the cognitive-more-capable regime.
%%Since, superposition coding is the optimal strategy again
%%this capacity result provides an alternative encoding

The rest of the paper is organized as follows.
The system model and definitions are presented in Section \ref{sec:models}.
Section \ref{sec:DM-CIC} provides the main result of this paper, which includes the capacity region for the cognitive-more-capable DM-CIC.
In Section \ref{sec:dis}, we show that the new capacity result includes all existing capacity results as subsets.
%Section \ref{sec:G} is devoted to the GCIC where we find its capacity at the cognitive-more-capable regime.
This is followed by conclusions in Section~\ref{sec:sum}.

%%%%%%%%%%%%%%%%%%%%%%%%%%%%%%%%%%%%%%%%%%%%%%%%%%%%%%%%%%%%%%%%%%%%%%%%%%%%%%%%%%%%%%%%%%%%%%%%%%%%%%%%%%%%%%%%%%%%%%%%%%%%%%%%%%%%%%%%%%%
%%%%%%%%%%%%%%%%%%%%%%%%%%%%%%%%%%%%%%%%%%%%%%%%%%%%%%%%%%%%%%%%%%%%%%%%%%%%%%%%%%%%%%%%%%%%%%%%%%%%%%%%%%%%%%%%%%%%%%%%%%%%%%%%%%%%%%%%%%%

\section{Preliminaries and Definitions}
\label{sec:models}

The two-user DM-CIC is an interference channel that consists of two
transmitter-receiver pairs, in which the {\it cognitive} transmitter
non-causally knows the message of the  {\it primary} user, in addition to its own message.
In what follows, we formally define DM-CIC and two special classes of that.

\subsection{Discrete memoryless cognitive interference channel}

\begin{figure}
\begin{center}
\scalebox {1}{
\begin{picture}(280,170)

\put(0,50){
\begin{picture}(200,80)
\put(20,20){\framebox(55,20){Encoder 2}}
\put(-10,30){\vector(1,0){30}}
\put(12,37){\makebox(0,0)[r]{$M_{2}$}}
%\put(-55,20){\framebox(45,20){Source 2}}
\put(75,30){\vector(1,0){30}}
\put(96,37){\makebox(0,0)[r]{$X^n_{2}$}}
\put(105,10){\framebox(40,90)}
\put (117,25){\rotatebox{90}{$p(y_1,y_2|x_1,x_2)$}}
\put(145,30){\vector(1,0){30}}
\put(165,37){\makebox(0,0)[r]{$Y^n_{2}$}}
\put(175,20){\framebox(55,20){Decoder 2}}
\put(230,30){\vector(1,0){30}}
\put(250,38){\makebox(0,0)[r]{$\hat{M}_{2}$}}
\end{picture}
}

\put(0,100){
\begin{picture}(200,80)
\put(20,20){\framebox(55,20){Encoder 1}}
\put(-10,30){\vector(1,0){30}}
\put(12,37){\makebox(0,0)[r]{$M_{1}$}}
%\put(-55,20){\framebox(45,20){Source 1}}
\put(75,30){\vector(1,0){30}}
\put(96,37){\makebox(0,0)[r]{$X^n_{1}$}}
%\put(105,10){\framebox(75,90){$p(y_1,y_2|x_1,x_2)$}}
\put(145,30){\vector(1,0){30}}
\put(165,37){\makebox(0,0)[r]{$Y^n_{1}$}}
\put(175,20){\framebox(55,20){Decoder 1}}
\put(230,30){\vector(1,0){30}}
\put(250,37){\makebox(0,0)[r]{$\hat{M}_{1}$}}
\put(5,30){\line(0,-1){25}} \put(5,5){\line(1,0){20}}
\put(25,5){\vector(0,-1){15}}
\end{picture}
}
\end{picture}
}
\end{center}

\vspace{-35pt}

\caption{The discrete memoryless cognitive
interference channel (DM-CIC) with two transmitters and two receivers.  $M_1, M_2$ are two messages,
$X_1, X_2$ are the inputs, $Y_1, Y_2$ are the outputs, and $p(y_1,y_2|x_1,x_2)$ is the transition probability of channel.}
  \label{fig:DM-CZIC}
\end{figure}
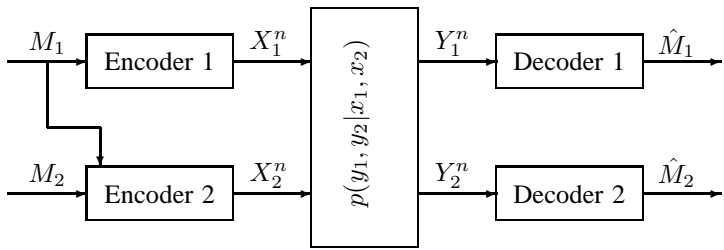

The DM-CIC is depicted in Fig.~\ref{fig:DM-CZIC}. Let $M_1$ and $M_2$
be two independent messages which are uniformly distributed on the sets of
all messages of the first and second users, respectively.  Transmitter $i$ wishes to transmit message $M_i$
to receiver $i$, in $n$ channel use at rate $R_i$, and $i=1,2$. Message $M_2$ is available only at transmitter 2, while both
transmitters know $M_1$. This channel is defined by a tuple $({\cal
X}_1,{\cal X}_2;p(y_1,y_2|x_1,x_2);{\cal Y}_1,{\cal Y}_2)$ where
${\cal X}_1,{\cal X}_2$ and ${\cal Y}_1,{\cal Y}_2$
are input and output alphabets, and $p(y_1,y_2|x_1,x_2)$ is channel transition probability density
functions.

%A special case of this channel, known as discrete memoryless cognitive
%Z-interference channel (DM-CZIC), arises when channel transition probability can be
%factorized as $p(y_1,y_2|x_1,x_2) = p(y_2|x_2)p(y_1|x_1,x_2)$. In such a case,
%interference is one sided. More precisely, the primary
%user does not interfere the secondary one. Similarly, if the interference
%link from the cognitive transmitter to the primary receiver does not exist a different DM-CZIC
%is realized; this has been studied in \cite{liu2009bounds}; we will not consider it
%in this paper.

The capacity of the DM-CIC is known in the ``cognitive less noisy'' \cite{vaezi2012lessnoisy}, ``strong interference'' \cite{Maric1},
``weak interference'' \cite{Wu-Vishwanath}, and ``better cognitive decoding'' \cite{Rini2} regimes.
These capacity results are listed in Table~\ref{table1}, and labeled $\mathcal C_I$, $\mathcal C_{II}$, $\mathcal C_{III}$,
and $\mathcal C_{III}^\prime$, respectively.
 In all above cases,
the cognitive receiver has a better condition (more information) than the primary one in some sense,
as it can be understood from the corresponding conditions in Table~\ref{table1}.
%In the second case, both receivers can decode both messages.
%%This is also correct in the first case.

\subsection{More Capable DM-CIC}

Since the second transmitter has complete and non-causal
knowledge of both messages, by sending the two messages, it can act like a {\it broadcast} transmitter.
Particularly, in the absence of the first transmitter this channel becomes the
well-known DM-BC \cite{Cover}. In the presence of the primary transmitter, this
channel is no longer a BC; however, similar to that in the DM-BC, one can define
conditions for which one of the receivers is in a ``better''
condition than the other one in decoding the messages, e.g., one receiver is {\it less noisy}
or {\it more capable} than the other \cite{ElGamal2011network}.

In \cite{vaezi2011capacity}, the authors extended this
notion to the DM-CIC, and studied
the case where the primary receiver is more capable than the cognitive receiver.
This led to the capacity of GCZIC at very strong interference.
In what follows, we show that similar to the less noisy DM-CIC \cite{vaezi2012lessnoisy},
and depending on which receiver is in the better condition than the other,
two different more capable DM-CIC arises. These two are formally defined in the following.

\begin{defn}
The DM-CIC is said to be {\it primary-more-capable}%\footnote{This definition first appeared in \cite{vaezi2011capacity}
%under the name of the ``more capable" DM-CIC.}
if
\begin{align}
I(X_1,X_2;Y_1) \geq I(X_1,X_2;Y_2)
\label{eq:defn1}
\end{align}
\label{cond1}
 for all $p(x_1,x_2)$.
\end{defn}

\begin{defn}
The DM-CIC is said to be {\it cognitive-more-capable} if
\begin{align}
I(X_1,X_2;Y_2) \geq I(X_1,X_2;Y_1)
\label{eq:defn2}
\end{align}
\label{cond2}
 for all $p(x_1,x_2)$.
\end{defn}

\noindent It can be noted that in the first case the primary
receiver has more information, about transmitted codewords, than the cognitive receiver whereas
the reverse is true in the second case.
Therefore, given the channel condition, a DM-CIC can be either in the
{\it primary-more-capable} or in the {\it cognitive-more-capable} regimes.
The former was studied in \cite{vaezi2011capacity}. In this paper,
we focus on the latter case.

%In the next section, we establish two inner bounds and an outer
% bound for the DM-CIC.
%Later, we will use these bounds to prove the capacity of the
%DM-CIC under specific conditions.

\section{Main Results}
\label{sec:DM-CIC}

In this section, we first introduce a new
outer bound on the capacity of the cognitive-more-capable DM-CIC.
We then find an alternative representation of this outer bound; the new representation
is the same as an achievable rate region for the DM-CIC which is based on superposition coding.
 Consequently, we establish the capacity region of the cognitive-more-capable DM-CIC in this section.

\subsection{New Outer Bounds}

The following provides an outer bound on the capacity of the cognitive-more-capable DM-CIC,
defined in \eqref{eq:defn2}.

\begin{thm}
Define $\mathcal{R}_o$ as the set of all rate pairs $(R_{1},R_{2})$ such that
\begin{subequations}\label{O1}
\begin{align}
R_1 &\leq I(U,X_1;Y_1), \label{O1:1}\\
R_1 + R_2 &\leq I(U,X_1;Y_1) + I(X_2;Y_2|U,X_1),\label{O1:2}\\
R_1 + R_2 &\leq I(X_1, X_2;Y_2),\label{O1:3}
\end{align}
\end{subequations}
for the probability distribution $p(u,x_1,x_2)p(y_1,y_2|x_1,x_2)$.
Then, for some $p(u,x_1,x_2)$, $\mathcal{R}_o$
provides an outer bound on the capacity region of the cognitive-more-capable
DM-CIC, defined by \eqref{cond2}.
\label{thm1}
\end{thm}

\begin{proof}
The proof is provided in Section~\ref{anx1}.
\end{proof}

\begin{thm}
Let $\mathcal{R}_o^\prime$ be the set of all rate pairs $(R_{1},R_{2})$ such that
\begin{subequations}\label{O2}
\begin{align}
R_1 &\leq I(U,X_1;Y_1), \label{O2:1}\\
R_2 &\leq  I(X_2;Y_2|U,X_1),\label{O2:2}\\
R_1 + R_2 &\leq I(X_1, X_2;Y_2),\label{O2:3}
\end{align}
\end{subequations}
for the probability distribution $p(u,x_1,x_2)p(y_1,y_2|x_1,x_2)$.
Then $\mathcal{R}_o^\prime \equiv \mathcal{R}_o$, that is, $\mathcal{R}_o^\prime$
gives another representation of $\mathcal{R}_o$ and makes an outer bound on the capacity region of the cognitive-more-capable
DM-CIC, for some $p(u,x_1,x_2)$.
\label{thm2}
\end{thm}

\begin{proof}
To prove this we consider the following two cases:\\
\textit{case 1:}  when \eqref{O1:2} is redundant in $\mathcal{R}_o$, i.e.,
\begin{align}
I(U,X_1;Y_1) + I(X_2;Y_2|U,X_1) \geq I(X_1, X_2;Y_2).
\label{ineq1}
\end{align}
\textit{case 2:}  when \eqref{O1:3} is redundant in $\mathcal{R}_o$, i.e.,
\begin{align}
I(U,X_1;Y_1) + I(X_2;Y_2|U,X_1) \leq I(X_1, X_2;Y_2).
\label{ineq2}
\end{align}

In the first case, we can see that \eqref{O2:2} becomes redundant also.
This is because the right-hand side of \eqref{O2:2} is greater than or equal to
the difference between the right-hand sides of \eqref{O2:3} and \eqref{O2:1} if
 \eqref{ineq1} holds. Consequently, the remaining constraints in
$\mathcal{R}_o$ and $\mathcal{R}_o^\prime$ are the same, and $\mathcal{R}_o^\prime \equiv \mathcal{R}_o$.

In the second case, it is obvious that the third inequality is redundant both in \eqref{O1} and \eqref{O2}.
Therefore, the set of constraints in Theorem~\ref{thm1} reduces to
\begin{subequations}\label{O3}
\begin{align}
R_2 &\leq I(U,X_2;Y_2), \label{O3:1}\\
R_1 + R_2 &\leq I(U,X_2;Y_2) + I(X_1;Y_1|U,X_2).\label{O3:2}
\end{align}
\end{subequations}
Similarly, the set of constraints in Theorem~\ref{thm2} reduces to
\begin{subequations}\label{O4}
\begin{align}
R_1 &\leq I(X_1;Y_1|U,X_2),\label{O4:1}\\
R_2 &\leq I(U,X_2;Y_2). \label{O4:2}
\end{align}
\end{subequations}
Let $\mathcal{R}_{o1}$ denote the union of all rate pairs $(R_{1},R_{2})$  that satisfy \eqref{O3:1}-\eqref{O3:2}
and $\mathcal{R}_{o1}^\prime$ be the union of all rate pairs $(R_{1},R_{2})$ that satisfy \eqref{O4:1}-\eqref{O4:2};
%In \cite {vaezi2012comments}, we prove that $\mathcal{R}_{o1} \equiv  \mathcal{R}_{o1}^{\prime}$.
we show that $\mathcal{R}_{o1} \equiv  \mathcal{R}_{o1}^{\prime}$.
Intuitively, the convex hull of these two regions is the same since the corner point of both regions, which remains after applying the convex hull operation,
are exactly the same.
More formally, using the same argument as El Gamal~\cite{ElGamalMoreCapable}, we can see
that any point on the boundary of  $\mathcal{R}_{o1}^\prime$ is also on the boundary of  $\mathcal{R}_{o1}$.
This is because $R_{2}$ can be thought of as the rate of common message that can be decoded at both receivers
while $R_{1}$ is the rate of the private message. Now $(R_{2}, R_{1}) \in \mathcal{R}_{o1}^{\prime} $
if and only if $(R_{2}-t, R_{1}+t) \in \mathcal{R}_{o1}^{\prime} $ for any
$0 \leq t \leq R_{2}$. In other words, the common rate $R_{2}$ can be
partly or wholly private. Thus region $\mathcal{R}_{o1}^{\prime}$
can be represented as $\mathcal{R}_{o1}$, i.e., $\mathcal{R}_{o1} \equiv  \mathcal{R}_{o1}^{\prime}$.\footnote{
Similarly, as stated in \cite[Chapter 5]{ElGamal} when proving the capacity of less noisy BC, the (convex hull of) region
\begin{align*}
R_2 &\leq I(U;Y_2), \\
R_1 + R_2 &\leq I(U;Y_2) + I(X_1;Y_1|U),
\end{align*}
is an alternative characterization of
\begin{align*}
R_1 &\leq I(X_1;Y_1|U),\\
R_2 &\leq I(U;Y_2),
\end{align*}
for some for some $p(u, x)$. By replacing $U$ with $(U,X_2)$ in these two regions
we will get $\mathcal{R}_{o1}$ and $\mathcal{R}_{o1}^\prime$, respectively.
}

Therefore, the proof of Theorem~\ref{thm2} is completed as in the both cases $\mathcal{R}_o^{\prime} \equiv  \mathcal{R}_o$.
\end{proof}

It should be indicated that the outer bounds in Theorem~\ref{thm1} and Theorem~\ref{thm2}
are valid only for the cognitive-more-capable DM-CIC,
defined in \eqref{eq:defn2}. However,
if we remove the third inequalities (i.e., \eqref{O1:3} and \eqref{O2:3}) in those sets of inequalities,
the remaining constraints in each set provide outer bounds for any DM-CIC.
Therefore, as a corollary of Theorem~\ref{thm2} we have
\begin{cor}
The set of all rate pairs $(R_1, R_2)$ such that
\begin{subequations}\label{Outer2}
\begin{align}
R_1 &\leq I(U,X_1;Y_1), \label{Osimple:1}\\
R_2 &\leq  I(X_2;Y_2|U,X_1),\label{Osimple:2}
\end{align}
\end{subequations}
for some $p(u,x_1,x_2)$ provides an outer bound on the capacity region of the DM-CIC.
\label{cor0}
\end{cor}

\noindent Corollary~\ref{cor0} gives a simpler and more tractable representation of the outer bound introduced in \cite[Theorem 3.2]{Wu-Vishwanath}.

We next provide an achievable rate regions for the DM-CIC.

\subsection{An Achievable Rate Region}
\label{inner}
The following theorem gives an achievable rate regions for the DM-CIC.

\begin{thm}
The union of rate regions given by
\begin{subequations}\label{I1}
\begin{align}
R_1 &\leq I(W,X_1;Y_1), \label{I1:1}\\
R_2 &\leq I(X_2;Y_2|W,X_1),\label{I1:2}\\
 R_1 + R_2 &\leq I(X_1,X_2;Y_2),\label{I1:3}
\end{align}
\end{subequations}
is achievable for the DM-CIC, where the union is over all probability distributions $p(w,x_1, x_2)$.
\label{thm3}
\end{thm}

\begin{proof}
%This achievable region uses superposition encoding at the
%cognitive transmitter, as the cognitive user knows the primary user's message.
The proof of Theorem~\ref{thm3} uses the superposition coding idea in which
$Y_1$ can only decode $M_1$
while $Y_2$ (the more capable receiver) is intended to decode both $M_1$ and $M_2$.
Considering the space of all codewords, one can
view the $(W, X_1)$ as {\it cloud centers}, and the $X_2$ as {\it satellites} \cite{Kramerbook}.
The decoding is based on joint typicality.
The details of the proof can be found in \cite{vaezi2012lessnoisy}.
\end{proof}

\begin{rem}
The achievable region in Theorem~\ref{thm3} is a subset of
the achievable region in \cite[Theorem 7]{Rini2}. This can be shown by setting
$U=U_{1c}, X_1= U_{1pb}, X_2 = U_{2c}=U_{2pb}$
and using the Fourier-Motzkin elimination to simplify the region.
Note that the indices 1 and 2 need to be swapped.
%Further details can be found in Section~\ref{anx2}.

\end{rem}

\subsection{The Capacity of the Cognitive-More-Capable DM-CIC}
\label{sec:cap}

The capacity region of the cognitive-more-capable DM-CIC is established immediately in light
of the outer bound in Theorem~\ref{thm2} and the inner bound in Theorem~\ref{thm3}. That is,
the region define by $\mathcal{R}_o^\prime$, or equivalently the rate region characterized
in Theorem~\ref{thm3}, gives the capacity region of the DM-CIC when \eqref{eq:defn2} holds.
\begin{thm}
For the cognitive-more-capable DM-CIC defined in \eqref{eq:defn2},
the capacity region is given by the set of all rate pairs $(R_1, R_2)$ such that
\begin{subequations}\label{C2}
\begin{align}
R_1 &\leq I(W,X_1;Y_1), \label{C2:1}\\
R_2 &\leq I(X_2;Y_2|W,X_1),\label{C2:2}\\
R_1 + R_2 &\leq I(X_1,X_2;Y_2),\label{C2:3}
\end{align}
\end{subequations}
for some $p(w,x_1, x_2)$.
\label{thm4}
\end{thm}

Theorem~\ref{thm4} gives the largest capacity region for the DM-CIC channel to date.
We prove this in the next section by showing that this capacity result contains all
existing capacity results of the DM-CIC channel as its subsets.

\begin{table*}[tb]
\renewcommand{\arraystretch}{1.1}
\caption{Summary of existing and new capacity results for the DM-CIC.
The subscripts 1 and 2, respectively, denote the primary and secondary (cognitive) users.\rlap{\textsuperscript{*}}} \label{table1}
\centering
\scalebox{.99}{
\begin{tabular}{|c|c|c|c|c|}
\hline
\bfseries Label &\bfseries DM-CIC class &\bfseries Condition & \bfseries Capacity region &  \bfseries Reference \\
\hline\hline\

$\mathcal C_I$ & \mbox {cognitive-less-noisy} & $ I(U;Y_1) \leq I(U;Y_2) $ & $ R_1 \leq I(U;Y_1) $  & \cite{vaezi2012lessnoisy} \\
$  $ &  & $ $ & $ R_2 \leq I(X_2;Y_2|U) $   &  \\
\hline

$\mathcal C_{II}$  & \mbox {strong interference} &  $I(X_1,X_2;Y_1) \leq I(X_1,X_2;Y_2)$ & $  R_1  \leq I(X_1;Y_1) $  &  \cite{Maric2} \\
$ $ & & $ I(X_2;Y_2|X_1) \leq I(X_2;Y_1|X_1) $ & $ R_2 \leq I(X_2;Y_2|X_1) $   &  \\
\hline

$\mathcal C_{III}$ & \mbox {weak interference} & $ I(X_1;Y_1) \leq I(X_1;Y_2) $ & $ R_1 \leq I(U,X_1;Y_1) $ & \cite{Wu-Vishwanath} \\
 $ $ & & $ I(U;Y_1|X_1) \leq I(U;Y_2|X_1)  $ & $ R_2 \leq I(X_2;Y_2|U,X_1) $  &  \\
\hline

$$  & & $ $ & $ R_1 \leq I(U,X_1;Y_1) $   &  \\
$ \mathcal C_{III}^\prime$ & \mbox {better-cognitive-decoding} & $  I(U,X_1;Y_1) \leq I(U,X_1;Y_2) $ & $ R_2 \leq I(X_2;Y_2|X_1) $ & \cite{Rini2} \\
$ $ & & $ $ & $ R_1 + R_2 \leq I(U,X_1;Y_1) + I(X_2;Y_2|U,X_1) $ &   \\
\hline

$$  & & $ $ & $ R_1 \leq I(U,X_1;Y_1) $   &  \\
$ \mathcal C_{IV}$ & \mbox {cognitive-more-capable} & $  I(X_1,X_2;Y_1) \leq I(X_1,X_2;Y_2) $ & $ R_2 \leq I(X_2;Y_2|U,X_1) $   & Theorem~\ref{thm4} \\
$ $ & & $ $ & $ R_1 + R_2 \leq I(X_1, X_2;Y_2) $ &   \\
\hline

\end{tabular}}
\begin{flushleft}
 \qquad  \qquad \scriptsize\textsuperscript{*} It should be emphasized that $\mathcal C_{III}^\prime \equiv \mathcal C_{III}$ \cite{vaezi2012comments} and $\mathcal C_I   \subseteq \mathcal C_{II} \subseteq \mathcal C_{III}  \subseteq \mathcal C_{IV}$.
\end{flushleft}
\end{table*}

\section{Comparison and Classification}
\label{sec:dis}
In this section, we compare the capacity region obtained in Theorem~\ref{thm4}
with all of the previously known capacity results for the DM-CIC.
For ease of comparison, these results are summarized in Table~\ref{table1}.
We show that the capacity region of the cognitive-more-capable DM-CIC contains
all other capacity regions listed in Table~\ref{table1}, as subsets.
We also clarify the relation between the other capacity results.
More precisely, we prove that
\begin{align}
\mathcal C_I  \subseteq \mathcal C_{II} \subseteq  \mathcal C_{III}  \equiv \mathcal C_{III}^\prime \subseteq \mathcal C_{IV}.
\label{eq:P}
\end{align}

\noindent We first observe that
\begin{align}
\mathcal C_I  \subseteq \mathcal C_{III}^\prime \subseteq \mathcal C_{IV}.
\label{eq:P1}
\end{align}
 This is evident by conditions corresponding to $\mathcal C_I, \mathcal C_{III}^\prime,  \mathcal C_{IV}$  in Table~\ref{table1}, because
\begin{align*}
I(U;Y_1) &\leq I( U;Y_2) \quad \forall \;  p(u) \\
\Rightarrow  \; I(U,X_1;Y_1) &\leq I( U,X_1;Y_2) \quad \forall \;  p(u,x_1) \\
\Rightarrow \;  I(X_2,X_1;Y_1) &\leq I(X_2,X_1;Y_2) \quad \forall \; p(u,x_1,x_2).
\end{align*}
We next prove that
\begin{align}
\mathcal C_I  \subseteq \mathcal C_{II}.
\label{eq:P2}
\end{align}

\noindent To show this we resort to a different representation of the ``strong interference'' condition, for which $\mathcal C_{II}$ hold.
From \cite[eq. (87)-(88)]{Maric2} we know that the DM-CIC is  in the ``strong interference'' regime if
\begin{subequations}\label{def3eq}
\begin{align}
I(X_2;Y_2|X_1) &\leq I(X_2;Y_1|X_1), \label{def3eq:first}\\
I(X_1, X_2;Y_1) &\leq I(X_1, X_2;Y_2),\label{def3eq:second}
\end{align}
for all $p(x_1,x_2)$.
\end{subequations}
Also, from \cite[eq. (8a)-(8b)]{vaezi2012comments} we know that this set of conditions is equivalent to
\begin{subequations}\label{MYK2}
\begin{align}
I(U;Y_1|X_1) &= I(U;Y_2|X_1), \label{MYK2:first}\\
I(X_1;Y_1) &\leq I(X_1;Y_2),\label{MYK2:second}
\end{align}
for all $p(u,x_1,x_2)$.
\end{subequations}

\begin{figure}[!tb]
\begin{center}
\scalebox {1.1}{
\begin{tikzpicture}
\draw[black, thick, rotate=0] (-4,2.7) rectangle (4,-2.7);
\draw[blue, thick, rotate=0] (0,0) ellipse (110pt and 65pt);
\draw[red, dashed, thick, scale=1.5,rotate=0] (0,0)ellipse (58pt and 32pt);
\draw[black, dotted, thick, scale=1.5,rotate=0] (0,0)ellipse (33pt and 8pt);
\draw[gray, densely dotted, thick, scale=1.5,rotate=0] (0,0)ellipse (45pt and 18pt);
\draw  (0,2.8)node[right, above] {DM-CIC};
\draw  (0,1.6)node[right, above] {{\color[rgb]{.1,0.1,1} cognitive-more-capable}};
\draw  (0,.85)node[right, above]  {{\color[rgb]{1,0.1,.1} better-cognitive-decoding}};
\draw  (0,-0.25)node[right, above] { cognitive-less-noisy} ;
\draw  (0,.35)node[right, above] {{\color[rgb]{.5,.5,.5} strong interference}} ;
\end{tikzpicture}
}
\end{center}
\caption{The class of the discrete memoryless cognitive interference channels (DM-CIC).
The cognitive receiver is {\it superior} than the primary receiver for the cognitive-more-capable and all its subclasses.
The largest ellipse (blue, solid line) represents the cognitive-more-capable DM-CIC. The ellipse with red, dashed lines
represents the better-cognitive-decoding DM-CIC; note that, this regime is equivalent to weak interference regime. The other
two ellipses, i.e., dotted and densely dotted ellipses, respectively show the cognitive-less-noisy and strong interference DM-CIC.}
\label{fig:CICclass}
\end{figure}
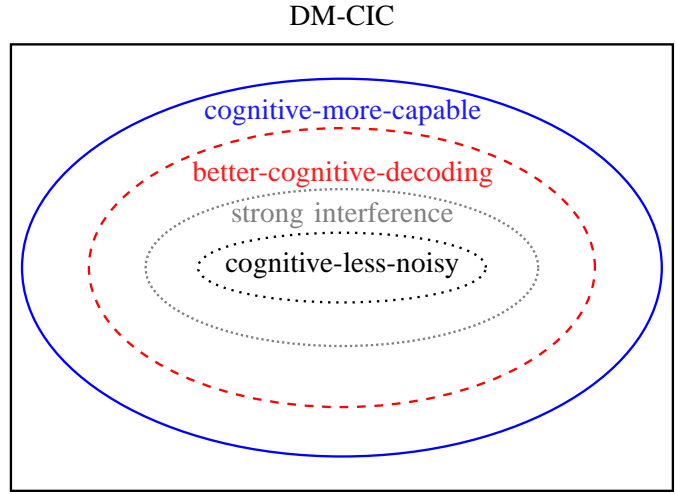

\noindent Now, in light of \eqref{MYK2}, it is straightforward to show that the condition required for the cognitive-less-noisy
regime, i.e.,
\begin{align}
 I(U;Y_1) \leq I(U;Y_2),
\label{eq:V}
 \end{align}
implies the conditions required for the strong interference regime.
To prove this we show that \eqref{eq:V} implies both \eqref{MYK2:first} and \eqref{MYK2:second}.
First, we see that if  $ I(U;Y_1) \leq I(U;Y_2)$ holds for any $p(u,x_1,x_2)$ then
we obtain $ I(U;Y_1|X_1) \leq I(U;Y_2|X_1)$ for all $p(u,x_1,x_2)$, thus $\eqref{eq:V} \Rightarrow \eqref{MYK2:first}$.
Also, for $U=X_2$ the condition in \eqref{eq:V} reduces \eqref{MYK2:second},
i.e., $\eqref{eq:V} \Rightarrow \eqref{MYK2:second}$. Hence, the condition required for the cognitive-less-noisy
regime implies that of the strong interference regime; this means that \eqref{eq:P2} holds.
Finally, by virtue of \eqref{eq:P1} and \eqref{eq:P2} and the fact that
$ \mathcal C_{II} \subseteq \mathcal C_{III}\equiv \mathcal C_{III}^\prime$
(see \cite[Claim 1]{vaezi2012comments}) it is obvious that \eqref{eq:P} is correct.
%The parts $\mathcal C_{III}  \equiv \mathcal C_{III}^\prime$ and $ \mathcal C_{II} \subseteq \mathcal C_{III}$  are proved in \cite{vaezi2012comments}.
In words, for a DM-CIC the followings are correct:
\begin{enumerate}
\item The better-cognitive-decoding and weak interference are equivalent (see \cite{vaezi2012comments}).
%\footnote{In \cite{vaezi2012comments} it is shown that the conditions required
%for a DM-CIC to be in the ``better-cognitive-decoding" and ``weak interference" regimes are equal.
%As a result, we can simplify representation of these two capacity results in the following form.
%\begin{pro}
%For a DM-CIC satisfying $I(U,X_1;Y_1) \leq I(U,X_1;Y_2)$,
%the capacity region is given by the set of all rate pairs $(R_1, R_2)$ such that
%\begin{subequations}\label{P2}
%\begin{align}
%R_1 &\leq I(U,X_1;Y_1), \label{P2:1}\\
%R_2 &\leq I(X_2;Y_2|U,X_1),\label{P2:2}
%\end{align}
%\end{subequations}
%for some $p(u,x_1, x_2)$.
%\label{pro1}
%\end{pro}
%}
%\item If a DM-CIC is in the strong interference regime then it is in the better-cognitive-decoding regime (see \cite{Rini2}).
\item If a DM-CIC is in the strong interference regime then it is in the cognitive-more-capable regime, as well.
\item If a DM-CIC is in the cognitive-less-noisy regime then it is in the strong interference regime.
\item A better-cognitive-decoding DM-CIC is cognitive-more-capable, too.
\end{enumerate}

Note that, the converse of statements 2, 3, and 4 does not hold in general.
Figure~\ref{fig:CICclass} represents these relations, pictorially.
In light of the above classifications, the constraints characterizing the
capacity region of the cognitive-more-capable DM-CIC (i.e., $\mathcal C_{IV}$)
can be use to represent the capacity region of all other classes of the DM-CIC listed in Table~\ref{table1}.
\begin{rem}
The constraints in Theorem~\ref{thm4} provide the capacity region of the DM-CIC
at the cognitive-less-noisy, strong interference, weak interference, better-cognitive-decoding, and cognitive-more-capable regimes.
Furthermore, when
the condition corresponding to each one of those subclasses holds, $\mathcal C_{IV}$ reduces to
the corresponding capacity result.\footnote{To better appreciate this, we may think of the two well-known
``less noisy" and ``more capable" BC and the relation between their capacity.
% in the following example.
%\begin{ex}
We know that the condition required for a less noisy BC implies that of more capable BC \cite{ElGamal2011network}.
That is, any less noisy BC is more capable also. Hence, the capacity of
more capable BC reduces to that of less noisy BC once the condition required for less noisy
BC is met. This is clear from the capacity regions of these two channels \cite{ElGamal2011network}.
%\end{ex}
}
For one thing, if a cognitive-more-capable DM-CIC further satisfies $I(U,X_1;Y_1) \leq I(U,X_1;Y_2)$
then it easy to see that the third constraint in $\mathcal C_{IV}$ becomes redundant, and the capacity region corresponding to
the better-cognitive-decoding, is achieved.
Hence $\mathcal C_{IV}$ unifies the representation of the capacity results for the DM-CIC  in different regimes.

\end{rem}

\begin{rem}
Superposition coding is optimal for several classes of the DM-CIC
for which the cognitive receiver is superior than the primary one, as we detailed in this section.
It is, however, not optimal in general since requiring the cognitive
receiver to recover both messages (even though non-uniquely for the primary's message)
can excessively constraint the achievable rate region.
\end{rem}

\begin{rem}
All of the classes defined in Table~\ref{table1} and depicted in Fig.~\ref{fig:CICclass},
imply the superiority of the cognitive receiver the primary one.
It is worth noting that, by swapping the indices
1 and 2 in the conditions, we can define similar classes in which the primary receiver
is superior than the cognitive one. One may expect similar capacity results in the new cases
by using superposition encoding in a different order. But it cannot come true because the factorization of probability
distribution $p(u,x_1, x_2)$ is different since the primary encoder
does not know $x_2$; thus, similar rate regions are not attainable over
general distribution $p(u,x_1, x_2)$.
This has been noted in \cite{vaezi2011capacity}.
\end{rem}

\section{Conclusions}
\label{sec:sum}

We have established the capacity of a new class of DM-CIC, named the cognitive-more-capable
DM-CIC, which gives the largest capacity region for the DM-CIC up to now.
This is proved by showing that all previously known capacity regions, for the DM-CIC,
are subsets of this new result, which is obtained by using superposition coding at the cognitive transmitter.
The analysis of the other capacity results of the DM-CIC shows that
superposition coding is the capacity-achieving techniques in those cases, too.
Besides, we  make a logical link between the different capacity results for this channel.
This sheds more light on the existing capacity results and unifies all of them under the
capacity region in the new regime.
%Finally, we find the capacity region for the cognitive-more-capable GCIC.
%This proves that, in the  cognitive channel, superposition coding can be the optimal coding strategy even at weak interference.

\section{Appendix}
\label{anx}

\subsection{Proof of Theorem \ref{thm2}}
\label{anx1}
The first two constraints in this outer bound (i.e., \eqref{O1:1} ,\eqref{O1:2}), which
make an outer bound on the capacity of any DM-CIC, are proved
in \cite[Theorem 3.2]{Wu-Vishwanath}. Here, we prove
that the last constraint \eqref{O1:3} holds for the cognitive-more-capable DM-CIC,
i.e., a DM-CIC that satisfies \eqref{cond2}.
%The proof follows the same footstep as the proof of the last inequality
%in \cite[Theorem 2]{vaezi2011superposition}; the only difference is that
%\cite[Theorem 2]{vaezi2011superposition} is proved for the primary-more-capable DM-CIC
%while we are dealing with the cognitive-more-capable DM-CIC.
%For completeness we write the details.
To do so, we can bound the rates $R_2 + R_2$ as
\begin{align}
 n(R_1 +  R_2)  &=  H(M_1, M_2) \nonumber \\
  & =  H(M_1| M_2)+ H(M_2) \nonumber \\
  & =  I(M_1;Y^n_1|M_2) + H(M_1|Y^n_1,M_2) \nonumber \\
  &\quad +  I(M_2;Y^n_2)  + H(M_2|Y^n_2) \nonumber \\
  & \leq I(M_1;Y^n_1|M_2) + I(M_2;Y^n_2)  + n\epsilon_{n} \label{eq:F-2}
 \end{align}
 where (\ref{eq:F-2}) follows by Fano's inequality.

Next, we bound the mutual information terms on the right-hand side of the inequality in (\ref{eq:F-2}).
%\small
\begin{subequations}\label{cnv}
\begin{align}
  I(&M_1;Y^n_1|M_2) + I(M_2;Y^n_2) \nonumber \\
  =& \sum^n_{\substack{i=1}}I(M_1; Y_{1i}|M_2,Y^{i-1}_{1}) + \sum^n_{\substack{i=1}}I(M_2;Y_{2i}|Y^n_{2,i+1}) \label{cnv:1} \\
  \leq& \sum^n_{\substack{i=1}}I(M_1, Y^n_{2,i+1}; Y_{1i}|M_2,Y^{i-1}_{1}) + \sum^n_{\substack{i=1}}I(M_2, Y^n_{2,i+1};Y_{2i})\nonumber \\
  =& \sum^n_{\substack{i=1}}I(M_1, Y^n_{2,i+1}; Y_{1i}|M_2,Y^{i-1}_{1})   \nonumber \\ &
  + \sum^n_{\substack{i=1}}I(M_2, Y^n_{2,i+1},Y^{i-1}_{1};Y_{2i}) - \sum^n_{\substack{i=1}}I(Y^{i-1}_{1};Y_{2i}|M_2,Y^n_{2,i+1})\nonumber \\
  =& \sum^n_{\substack{i=1}}I(M_1;  Y_{1i}|M_2,Y^{i-1}_{1},Y^n_{2,i+1}) +
   \sum^n_{\substack{i=1}}I(M_2, Y^n_{2,i+1},Y^{i-1}_{1};Y_{2i}) \nonumber \\ & - \sum^n_{\substack{i=1}}I(Y^{i-1}_{1};Y_{2i},|M_2,Y^n_{2,i+1}) %\nonumber \\ &
   + \sum^n_{\substack{i=1}}I(Y^n_{2,i+1};Y_{1i},|M_2,Y^{i-1}_{1}) \nonumber \\
  =& \sum^n_{\substack{i=1}}I(M_1; Y_{1i}|V_{i}) +  \sum^n_{\substack{i=1}}I(V_i;Y_{2i})  \label{cnv:4} \\
   \leq& \sum^n_{\substack{i=1}}I(X_{1i},X_{2i}; Y_{1i}|V_{i}) + \sum^n_{\substack{i=1}}I(V_i;Y_{2i}) + n\epsilon_{n} \label{cnv:5} \\
     \leq& \sum^n_{\substack{i=1}}I(X_{1i},X_{2i}; Y_{2i}|V_{i}) + \sum^n_{\substack{i=1}}I(V_i;Y_{2i}) + n\epsilon_{n} \label{cnv:6} \\
       % \label{eq:b1}
   = & \sum^n_{\substack{i=1}}I(X_{1i},X_{2i}; Y_{2i})+ n\epsilon_{n} \nonumber
\end{align}
\end{subequations}

\normalsize
\noindent in which \eqref{cnv:1}  follows by the chain rule;
\eqref{cnv:4}  follows by the Csiszar sum identity and the
auxiliary random variable $V_i \triangleq (M_2, Y^{i-1}_{1}, Y^n_{2,i+1})$;
\eqref{cnv:5} follows from  $M_1 \rightarrow (X_1, X_2)
\rightarrow Y_1$; \eqref{cnv:6} follows from the cognitive-more-capable condition in \eqref{eq:defn2} that gives
$I(X_1,X_2;Y_2) \geq I(X_1,X_2;Y_1)$, and implies that
$I(X_1,X_2; Y_1|V) \leq I(X_1,X_2; Y_2|V)$.

Next, we define the time sharing random variable $Q$
which is uniformly distributed over $[1:n]$ and is independent of
$(M_1,M_2,X^n_1,X^n_2,Y^n_1,Y^n_2)$. Also we define
$X^n_1 = X^n_{1Q}$, $X^n_2 =X^n_{2Q}$, and $Y^n_2=Y^n_{2Q}$.
 Then we have
\begin{align*}
 n(R_1 +  R_2)  &\leq  \sum^n_{\substack{i=1}}I(X_{1i},X_{2i}; Y_{2i})+ n\epsilon_{n} \\
  & =nI(X_1,X_2;Y_2|Q) + n\epsilon_{n}\\
   &\leq nI(X_1,X_2;Y_2) + n\epsilon_{n}.  \label{eq:Rsum1}
 \end{align*}

\noindent But $\epsilon_{n} \rightarrow 0$, as $n \rightarrow \infty $,
because the probability of error is assumed to vanish.
This completes the proof.

\end{document}